\documentclass[5p,twocolumn]{elsarticle} 

\journal{Automatica}
\biboptions{numbers,sort&compress}

\usepackage{natbib}   

\usepackage{xcolor}

\usepackage{amsfonts,latexsym,amsmath,amssymb}
\usepackage[mathscr]{eucal}
\usepackage{graphicx,color}
\usepackage{comment}
\usepackage{hyperref}
\usepackage[utf8]{inputenc}
\usepackage{mathtools}
\usepackage{hyperref}
\newtheorem{theorem}{Theorem}

\newtheorem{proposition}{Proposition}

\newtheorem{assumption}{Assumption}
\newtheorem{remark}{Remark}

\usepackage{lipsum}
\usepackage{mathtools}
\usepackage{cuted}
\usepackage{xifthen}

\catcode`\@=11
\def\downparenfill{$\m@th\braceld\leaders\vrule\hfill\bracerd$}
\def\overparen#1{\mathop{\vbox{\ialign{##\crcr\crcr
\noalign{\kern0.4ex}
\downparenfill\crcr\noalign{\kern0.4ex\nointerlineskip}
$\hfil\displaystyle{#1}\hfil$\crcr}}}\limits}
\catcode`\@=12

\def\NN{{\mathbb N}}    
\def\RR{{\mathbb R}}    
\def\BB{{\mathcal B}}

\DeclareMathOperator{\dis}{d}

\def\P{{\Pi}}
\def\upC{{\mathcal{\overline C}}}
\def\lowC{{\mathcal{\underline C}}}

\def\cA{{\mathcal{A}}}
\def\cD{{\mathcal{D}}}
\newcommand\cV[3][]{\ifthenelse{\isempty{#1}}{\mathcal{V}_{#2} (#3)}{\mathcal{V}_{#2}^{#1} (#3)}}
\def\bV{{\boldsymbol{V}}}

\def\cX{{\mathcal{X}}}

\usepackage{verbatim}

\begin{document}

\begin{frontmatter}

\title{Total stability and integral action for discrete-time nonlinear systems}

\tnotetext[t1]{Research partially funded by the ANR Delicio project.}

\author[LAGEP]{S. Zoboli}
\ead{samuele.zoboli@univ-lyon1.fr}
\author[LAGEP]{D. Astolfi}
\ead{daniele.astolfi@univ-lyon1.fr}
\author[LAGEP]{V. Andrieu}
\ead{vincent.andrieu@univ-lyon1.fr}

\address[LAGEP]{Universit\'e Lyon 1, Villeurbanne, France -- CNRS, UMR 5007, LAGEPP, France. }

\begin{abstract}
Robustness guarantees are important properties to be looked for during control design. They ensure stability of closed-loop systems in face of uncertainties, unmodeled effects and bounded disturbances. While the theory on robust stability is well established in the continuous-time nonlinear framework, the same cannot be stated for its discrete-time counterpart. In this paper, we propose the discrete-time parallel of total stability results for continuous-time nonlinear system. This enables the analysis of robustness properties via simple model difference in the discrete-time context. First, we study how existence of equilibria for a nominal model transfers to sufficiently similar ones. Then, we provide results on the propagation of stability guarantees to perturbed systems. Finally, we relate such properties to the constant output regulation problem by motivating the use of discrete-time integral action in discrete-time nonlinear systems.
\end{abstract}


\end{frontmatter}

\section{Introduction}
\label{sec:intro}
Stability of an equilibrium is of utter importance in the design of feedback controllers. Stability properties are typically inferred through the analysis of the system model in closed-loop. However, it is well known that uncertainties inherently exist in control applications, due to the presence of unmodeled effects and parameters mismatches. To tackle the issue, robust control and robust stability analysis have become fundamental tools for control design \cite{zhou1998essentials}. In practice, due to the nature of digital controllers, these laws are implemented in a discrete-time form. The problem of synthesizing such robust controllers can be cast in an optimization framework, both in the linear and the nonlinear scenario, e.g. \cite{petersen1987stabilization,petersen1986riccati,de1999new, kiumarsi2017h, li2020adaptive}. Nevertheless, a strong theoretical foundation turns out to be a fundamental design tool.  While the theory on robust stability and regulation is well-developed and mature for continuous-time systems, see e.g., \cite{francis1976internal,astolfi2017integral}, the discrete-time nonlinear scenario still misses some important results. Although it is known that, for sufficiently small sampling times, continuous-time results are valid in the discrete framework, such an approach may be restrictive or inapplicable for some control applications. 

In the context of discrete-time nonlinear systems,  some necessary local conditions linked to robustness appeared in \cite{lin1994design}. Here, the authors state that  a necessary condition for local stability of nonlinear discrete-time autonomous systems comes from the solvability of a  nonlinear equation. Also, results on robust  stability  appeared in \cite[Chapter 5]{agarwal2000difference}. Therein, it is shown that if the origin of the nominal system is locally stable and Lipschitz, then it is also locally robustly stable for bounded disturbances. More recently, in the context of converse Lyapunov theorems for discrete-time systems, \cite{kellett2004smooth,kellett2005robustness} proved that a necessary condition for robust stability is the existence of a smooth Lyapunov function.  Yet, to the best of authors' knowledge, there are no results mimicking well-established general results of ``total stability'' for continuous-time systems. The concept of total stability was firstly introduced in the works
of Dubosin \cite{dubosin1940problem}, Gorsin \cite{gorsin1948stability}, Malkin
\cite{malkin1944stability},   
and more recently studied in 
\cite{yoshizawa1969relationship,salvadori1977problem,astolfi2017integral}. In this context, 
robustness properties are analyzed directly via nonlinear unstructured models differences. This allows inferring 
the preservation of a stable equilibrium point for sufficiently similar plants by means of simple model comparison
\cite[Lemma 4, 5]{astolfi2017integral}.

The goal of this paper is therefore to translate such results to the discrete-time scenario. We draw conclusions similar to the continuous case, yet under some fundamental differences, given by the discrete nature of the system. In particular, we show that stability properties of the equilibrium of a nominal model imply the existence and stability of an equilibrium (possibly different from the former) for any perturbed system sufficiently ``close'' to the nominal one. The result is proved under some regularity assumptions and bounded mismatches. Moreover, we provide a counterexample highlighting that some results from the continuous-time scenario may not apply in the discretized framework. This disproves some arguments of \cite{lin1994design}. Finally, we link the obtained results on robust stability to the robust output regulation problem, building on recent forwarding techniques \cite{mattioni2017lyapunov,mattioni2019forwarding}. We justify the addition of an integral action for rejecting constant disturbances or tracking constant references. More specifically, we show that if the true model of the plant to be controlled is sufficiently close to one used for controller design, then output regulation is still achieved.

The rest of the paper is organized as follows: Section~\ref{sec:total_stab} presents the main result of the paper on total stability; Section~\ref{sec:regulation} applies the result to the problem of constant output regulation; 
Section~\ref{sec:conclusion} comments and concludes the paper.

\subsection{Notation} \label{sec:notation}

$\RR$, resp. $\NN$, denotes the set of 
real numbers, resp. nonnegative integers. $\RR_{\ge0}$ denotes the set of nonnegative real numbers.
In this work, we define a time-invariant nonlinear
dynamical discrete-time system as 
$x^+=f(x)$, where $x\in\RR^n$ is the state solution evaluated at timestep $k\in \NN$ with initial condition $x_0$, and $x^+$  is the state solution at step $k+1$.
Sets are denoted by calligraphic letters and, for a given set $\cX$, we identify its boundary by $\partial \cX$. The notation $\cX \setminus \mathcal G$ identifies the intersection between $\cX$ and the complement of $\mathcal G$. When a set $\cX$ is strictly included in a set $\mathcal G$, we use the notation $\cX \subsetneqq \mathcal G$. We use $|\cdot|$ as the norm operator for matrices and vectors. Moreover, we denote as  $\dis(x,\cX)$ a generic distance function between any point $x\in \RR^n$ and a closed set $\cX\subset \RR^n$. For instance, one may choose $\dis(x,\cX)= \inf_{z\in\cX}|x-z|$. Given a function $\alpha:\RR_{\ge 0}\to\RR_{\ge0}$ we say $\alpha \in \mathcal{K}$ if it is continuous, zero at zero and strictly increasing. Similarly, we say $\alpha \in \mathcal{K}_\infty$ if $\alpha \in \mathcal{K}$ and $\lim_{s\to\infty} \alpha(s)=\infty$. Moreover, for a square matrix $A$ we denote by $\lambda_{\max}(A)$ its maximum eigenvalue.
 As a final note, let $P\in \RR^{n\times n}$  be a symmetric positive definite matrix.
 For any $A \in \RR^{n\times n}$ and   for an arbitrary scalar $r:={r}_1 {r}_2$ with $r_1,{r}_2>0$, the generalized Schur's complement implies
 that the following inequalities are equivalent
$$
A^\top P A - r P \preceq 0 
\quad \Longleftrightarrow
\quad  
\begin{pmatrix}
-{r}_1 P & &A^\top P 
\\
P A & & -{r}_2 P
\end{pmatrix}
\preceq 0.
$$





    


\section{Total stability results}\label{sec:total_stab}

In this section, we study how the stability properties of 
the origin of a given discrete-time autonomous nonlinear system 
\begin{equation}
    \label{eqn:sys}
    x^+ = f(x) ,
\end{equation}
transfer to systems described by a sufficiently similar difference equation
\begin{equation}
    \label{eqn:pert_sys}
    x^+ = \hat f(x) ,
\end{equation}
where $f:\RR^n\to\RR^n \; ,\hat f :\RR^n\to\RR^n$ are continuous functions. We propose two different results. The first one links the existence of an equilibrium for system \eqref{eqn:sys} to the existence
of an equilibrium for the perturbed system 
\eqref{eqn:pert_sys}. 
We show that an equilibrium for \eqref{eqn:pert_sys} exists, provided that the two models are locally close enough. More precisely, the result holds if
the functions $f$ and $\hat f$ are not too different in the $C^0$ norm,
and system \eqref{eqn:sys} presents an attractive forward invariant set containing its equilibrium and which is homeomorphic to the unit ball. The second result considers the case where both the dynamics and the Jacobians of the two systems
\eqref{eqn:sys}, \eqref{eqn:pert_sys}
are sufficiently similar. Under such conditions, we show that the existence of a locally exponentially stable equilibrium for \eqref{eqn:sys} implies the existence of a locally exponentially stable equilibrium for \eqref{eqn:pert_sys} close to it. Moreover, we present a lower bound on the size of the domain of attraction of the equilibrium for \eqref{eqn:pert_sys}.

\subsection{Existence of equilibrium}

We now present the minimal assumption required to show the existence of an equilibrium for \eqref{eqn:pert_sys}. 
To this end, we introduce the following notation.
Given a positive function $V:\cA\subseteq\RR^n\to\RR_{\ge 0 }$ and a 
positive real number $c>0$, we denote the sublevel set of such a function as
 \begin{equation}\label{eq_SublevelSet}
         \cV{c}{V}:=\{x\in \cA:V(x)\leq c\}.
\end{equation}

\begin{assumption}
\label{asmp:homeo_sphere}
Let $\cA$ be an open subset of $\RR^n$.
There exists a $C^0$  function  $V:\cA\to\RR_{\ge 0 }$  satisfying $V(0)=0$ and such that the following holds:
\begin{enumerate}
    \item there exists a positive real number $\bar c$ such that the set $\cV{\bar c}{V}$ is homeomorphic to the unit ball;
\item there exists $\rho \in (0,1)$ such that 
    \begin{equation}\label{eqn:thm1_smooth_lyap}
    V(f(x)) \le \rho V(x), \qquad\forall\, 
    x\in \cV{\bar c}{V}.
\end{equation}
\end{enumerate}
\end{assumption}
 \begin{remark}\label{rem1} Assumption~\ref{asmp:homeo_sphere} is not requiring the nominal system \eqref{eqn:sys} to be asymptotically stable. It solely assumes the existence of an attractive forward invariant compact set which is homeomorphic to the unit ball.  However, since $V$ is not strictly positive outside of the origin,  
 $V$ may have local minima and does not allow 
 to conclude asymptotic stability of the origin.
 \end{remark}

The first result is formalized by the following proposition. More comments on the  assumption are postponed after the proof.

\begin{proposition} \label{prop:existence_eq}
Let Assumption~\ref{asmp:homeo_sphere} hold.
Then, for any positive $\underline c\leq  \bar c$ there exists a  positive real number $\delta$ such that, for any  continuous function $\hat f:\RR^n\to\RR^n$ 
satisfying
\begin{equation}\label{eq:prop1_condition}
\begin{array}{l}
    |\hat f(x)-  f(x)|< \delta,  \qquad \forall x \in \cV{\bar c}{V} 
\end{array}
\end{equation}
the corresponding
system \eqref{eqn:pert_sys}
 admits an equilibrium point $x_e \in \cV{\underline c}{V}$. Moreover, such systems has no other 
 equilibrium in the set
 $\cV{\bar c}{V}\setminus\cV{\underline c}{V}$.
\end{proposition}

\begin{proof}
Consider $\underline  c\leq \bar c$ and let $\tilde \rho$ be any positive real number satisfying
$$
\rho < \tilde \rho < 1\ .
$$
Since the set  $\cV{\bar c}{V}$ is homeomorphic to the unit ball, it is bounded. Moreover, the function $V$ being continuous, $\cV{\underline  c}{V}$ is a compact subset. 
 Next, we define the function $p:\RR_{\geq 0}\to\RR$ as
\begin{equation}
p(s) =\max_{\substack{x\in\cV{\bar c}{V}\\v \in \RR^n: |v|=1}}
  \big\{V(f(x)+sv)-\bar c\big\},
\end{equation}
with $s$ a positive real number.
Recalling item 2 of Assumption~\ref{asmp:homeo_sphere}, we obtain
\begin{align}
    p(0) &= \max_{\substack{x\in\cV{\bar c}{V}\\v \in \RR^n: |v|=1}}
  \big\{V(f(x))-\bar c\big\} \notag \\
  &\leq \max_{\substack{x\in\cV{\bar c}{V}\\v \in \RR^n: |v|=1}}
  \big\{ \rho V(x)-\bar c \big\} \notag \\
  &\leq (\rho-1)\bar c<0.\label{eq:F0}
\end{align}
Then, we define the function 
$q:\RR_{\geq 0}\to\RR$ as
$$
q(s) = 
\max_{\substack{x \in \cV{\bar c}{V}\setminus\cV{\underline c}{V}\\ 
v \in \RR^n: |v|=1}}
\big\{V(f(x)+sv)-\tilde \rho V(x) \big\}.
$$
It satisfies
\begin{align}\notag
q(0) &= \max_{x \in \cV{\bar c}{V}\setminus\cV{\underline c}{V}}\{V(f(x))-\tilde \rho V(x)\} ,\\ \notag
&\leq \max_{x \in \cV{\bar c}{V}\setminus\cV{\underline c}{V}}\{(\rho-\tilde \rho) V(x)\} ,
\\&\leq(\rho-\tilde \rho) \underline c<0.
\label{eq:G0}
\end{align}
As a consequence, since $p,q$ are continuous functions satisfying
\eqref{eq:F0}, \eqref{eq:G0}, there exists $\delta>0$ such that 
\begin{equation}
p(s)<0, 
\quad q(s)<0, 
\qquad \forall\, s\in[0,\delta].
\label{eq:pqnegative}
\end{equation}
Now, pick any continuous function $\hat f$ satisfying \eqref{eq:prop1_condition}.
Note that for all $x$ in $\cV{\bar c}{V}$ such that $f(x)\neq\hat f(x)$, 
\begin{equation}
    V(\hat f(x))-\bar c = V(f(x)+sv)-\bar c
\end{equation}
with $s=|f(x)-\hat f(x)|$, $v=\frac{f(x)-\hat f(x)}{|f(x)-\hat f(x)|}$. Consequently, 
 for all $x$ in $\cV{\bar c}{V}$ inequalities \eqref{eq:pqnegative} and \eqref{eq:prop1_condition} imply
\begin{equation}
    V(\hat f(x))-\bar c \leq p(|f(x)-\hat f(x)|)< 0.
\end{equation}
Hence, $\hat f(x)\in \cV{\bar c}{V}$.
By assumption, $\cV{\bar c}{V}$ is homeomorphic to a
unitary ball, which we denote as $\BB=\{z\in \RR^n :|z|\leq1\}$. Hence, there exists two continuous mappings $T:\cV{\bar c}{V}\to\BB$ and $T^{-1}:\BB\to\cV{\bar c}{V}$ such that $T\circ T^{-1}(z)=z$ for all $z$ in $\BB$.
Hence, the mapping
$$
T\circ\hat f\circ T^{-1} : \BB\to \BB
$$
 is a continuous function. Hence, by Brouwer's fixed point theorem, there exists $z^*\in \BB$ such that
$$
T\circ\hat f\circ T^{-1}(z^*)=z^*.
$$
Thus, it implies
$$
\hat f\circ T^{-1}(z^*) = T^{-1}(z^*).
$$
We deduce that 
$x_e=T^{-1}(z^*)$ is a fixed point  belonging to $\cV{\bar c}{V}$.
Now, let us consider the set $\cV{\bar c}{V}\setminus\cV{\underline c}{V}$. As before, with the same definitions of $s$ and $v$, inequalities \eqref{eq:pqnegative} and \eqref{eq:prop1_condition} imply that, for all $x \in \cV{\bar c}{V}\setminus\cV{\underline c}{V}$, it holds
\begin{align*}
V(\hat f(x)) - \tilde \rho V(x)&= V(f(x)+sv)- \tilde \rho V(x) \\
&\leq q(s)   <0.
\end{align*}
Hence, $\hat f(x)\neq x$ for all $x \in \cV{\bar c}{V}\setminus\cV{\underline c}{V}$. Consequently,  $x_e$ belongs to $\cV{\underline c}{V}$ and this concludes the proof.
\end{proof}

\subsection{About Assumption~\ref{asmp:homeo_sphere}}
\label{sec_topology}
Proposition~\ref{prop:existence_eq}
establishes the conditions under which the existence of an equilibrium for the nominal model \eqref{eqn:sys} implies the existence of an equilibrium
for any perturbed system
 \eqref{eqn:pert_sys}
sufficiently close to the nominal one. This result parallels the continuous-time one presented in \cite[Lemma 4]{astolfi2017integral}.
However, different assumptions are needed.
In the continuous-time case, 
the origin of the nominal system is supposed
to be asymptotically stable.
In turn, 
for any forward invariant set, 
this ensures the existence of 
a Lyapunov function whose  level sets 
are homeomorphic to a sphere, see
\cite[Theorem 1.2]{wilson1967structure}.
Unfortunately, in the discrete-time case 
this fact is in general not true. As a matter of fact,
Lyapunov level sets may be  
non-homeomorphic to spheres, contrarily to what is stated in \cite[Proof of Theorem 2.7]{lin1994design}.
As an example, see 
\cite{gonzaga2012stability}. 
 This phenomenon is due to the nature of such systems, as the presence of jumps doesn't allow an easy translation of continuous-time results.
As a consequence, in Assumption~\ref{asmp:homeo_sphere}
we ask for the existence of an invariant 
sublevel set $\cV{\bar c}{V}$ homeomorphic to a ball.
At the same time, we do not require asymptotic
stability of the origin, see Remark~\ref{rem1}.

 In the following, we present an example showing the aforementioned behavior. In particular, we focus on a simple and stable linear system . We carefully craft a non-trivial Lyapunov function guaranteeing asymptotic stability of the origin. Successively, we exploit the structure of such a function to show that it has no sublevel set which is homeomorphic to a ball.
 
Consider the linear system 
\begin{equation}\label{eq_SystLyap}
x^+ = f(x) = \frac{1}{2} I_n x,
\end{equation}
where $x\in \RR^n$, $I_n$ is the identity matrix of dimension $n$ and the candidate Lyapunov function
\begin{equation*}
V(x) =
\begin{cases}
0 & x=0\\
6|x| - 5 \cdot 2^i & 2^i \le|x| < 2^i+2^{i-1}\\
-4|x| +5\cdot2^{i+1} & 2^i+2^{i-1} \le |x| < 2^{i+1}\\
\end{cases}
\end{equation*}
where $i\in \mathbb Z \cup \{-\infty,+\infty\}$ and it is uniquely defined by $|x|$. It can be verified that 
$V(x)$ is continuous.
Consider the case where $2^i \le|x| < 2^i+2^{i-1}$ for some $i$ in $\mathbb Z$. We have $2^{i-1} \le |f(x)| < 2^{i-1}+2^{i-2}$ and
\begin{align*}
    V(f(x)) -V(x) &= 6|f(x)| - 5 \cdot 2^{i-1} - 6|x| + 5 \cdot 2^i \\
     &= -3|x| + 5 \cdot 2^{i-1}\\
     &< - 2^{i-1} < 0.
\end{align*}
Similarly, for $2^i+2^{i-1} \le |x| < 2^{i+1}$ for some $i$ in $\mathbb Z$ we obtain $2^{i-1}+2^{i-2} \le |x| < 2^{i}$ and
\begin{align*}
    V(f(x)) -V(x) &= -4|f(x)| +5\cdot2^{i} +4|x| -5\cdot2^{i+1} \\
     &= 2|x| - 5 \cdot 2^{i}\\
     &< - 2^{i} < 0.
\end{align*}
Hence, $V$ is a continuous Lyapunov function for the system. However,  it does not exist any $c>0$ such that the corresponding  sublevel set $\cV{c}{V}$ defined in \eqref{eq_SublevelSet}  is homeomorphic to a ball, since each sublevel set is not path-connected.
Figure~\ref{fig:counterex} shows such a behavior for the planar case $x \in \RR^2$. 

\begin{figure}
    \centering
    \includegraphics[width=\linewidth]{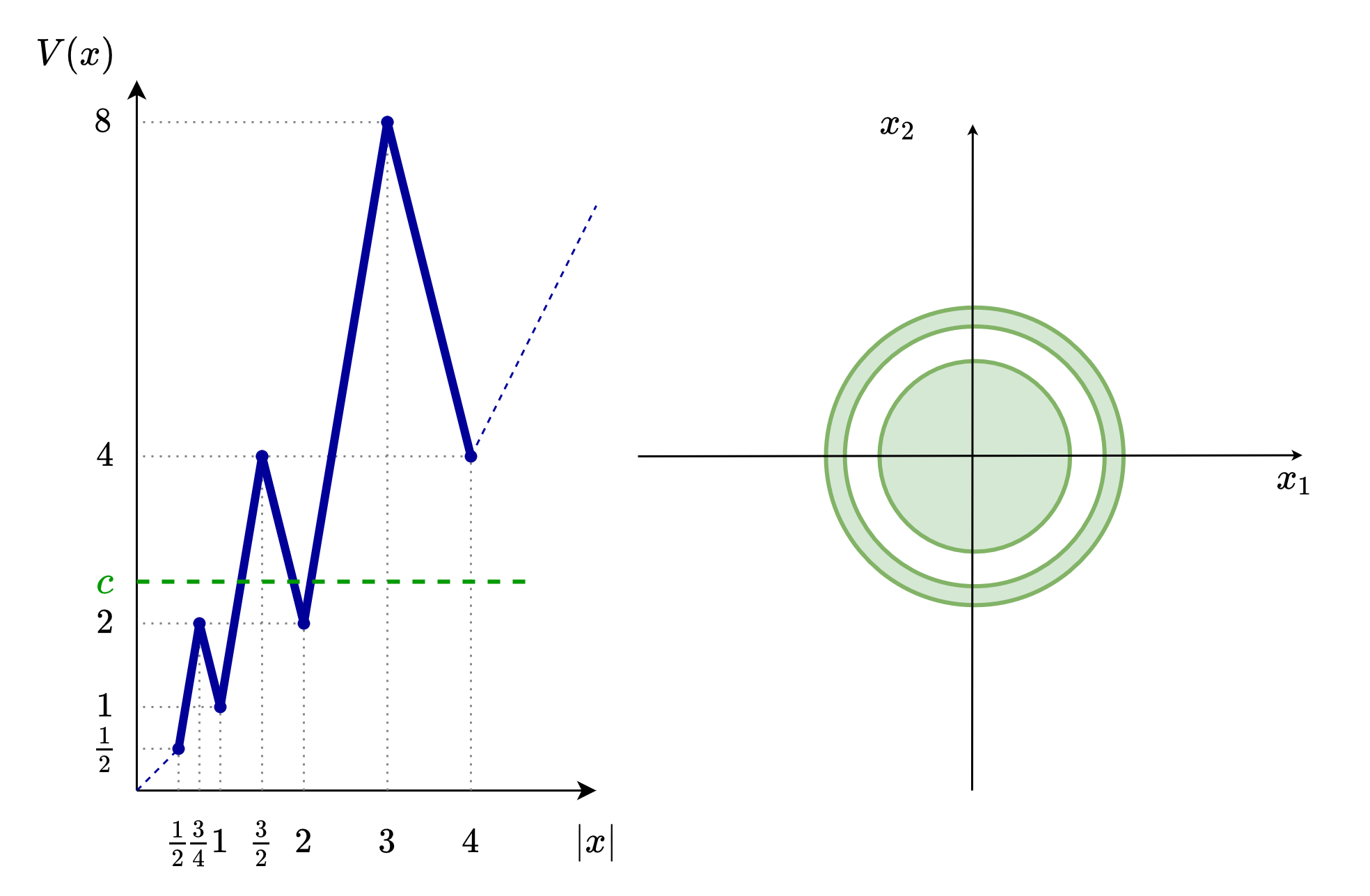}
    \caption{Lyapunov sublevel set for $x\in \RR^2$}
    \label{fig:counterex}
\end{figure}

Clearly, one could have picked a quadratic Lyapunov function for such a linear system. Yet, this example shows that discrete-time Lyapunov functions do not always possess the desired homeomorphicity property. As a consequence, we cannot guarantee that Lyapunov functions provided by converse Lyapunov theorems (e.g., \cite{jiang2002converse})  satisfy such a property.

\subsection{Existence of an exponentially stable equilibrium}

 We now present the main result on total stability
 of this paper, which is formalized in the next theorem. This result parallels \cite[Lemma 5]{astolfi2017integral}. We exploit the asymptotic properties of the equilibrium for the nominal system \eqref{eqn:sys} to prove existence, uniqueness and stability of an equilibrium for \eqref{eqn:pert_sys}.
 
\begin{theorem}\label{thm:perturbed_stable}
Assume the origin of the system \eqref{eqn:sys} 
is asymptotically stable with domain of attraction $\cA$ and locally exponentially stable. Let $\lowC$ be an arbitrary compact set satisfying $\{0\}\subsetneqq \lowC \subsetneqq \cA$ and 
suppose the function $f$ is
$C^0$ for all $x\in \cA$
 and moreover $C^1$ for all  $x\in \lowC$.
 Then, for any  forward invariant  compact set $\upC$ verifying
$$
\{0\} \subsetneqq   \lowC \subsetneqq   \upC  \subsetneqq  \cA,
$$
there exist a  positive scalar $\delta>0$ such that, for any function $\hat f$ which is
$C^0$ for all $x\in \cA$
and $C^1$ for all  $x\in \lowC$ and  and which satisfies
\begin{equation} \label{eqn:dyn_cond}
    |\hat f(x)-  f(x)|\leq \delta, \qquad \forall x \in \upC,
\end{equation}
\begin{equation}\label{eqn:jcb_cond}
    \left|
\dfrac{\partial \hat f}{\partial x} (x)
-
\dfrac{ \partial  f}{\partial x} (x)
\right| \leq \delta , \qquad \forall x \in \lowC,
\end{equation}
the corresponding  system \eqref{eqn:pert_sys} 
admits an equilibrium point $x_e\in \lowC$,  which is asymptotically stable 
with a domain of attraction containing $\upC$
and locally exponentially stable.
\end{theorem}

\begin{remark}
Differently from \cite[Lemma 5]{astolfi2017integral}, the nominal and perturbed models are required to be $C^1$  inside $\lowC$ solely, while $C^0$ outside. This allows  considering interesting continuous functions, such as saturation functions. 
\end{remark}

\begin{proof}
The proof is organized in several steps. First, we highlight some implications of  system \eqref{eqn:sys} being locally exponentially stable. Indeed, by converse Lyapunov theorem, the quadratic Lyapunov function for the linearized system around the origin is a valid local Lyapunov function for the nonlinear system. Hence, being the level sets of a quadratic function homeomorphic to a sphere, Proposition~\ref{prop:existence_eq} guarantees the existence of a sufficiently small bound $\delta$ such that there exists an equilibrium point $x_e$ for system \eqref{eqn:pert_sys}, arbitrarily close to the origin. However, we want to show its uniqueness and we want to provide more explicit bounds for such a model mismatch. Hence, the first three steps explicitly compute suitable bounds guaranteeing the existence, uniqueness and local exponential stability of an equilibrium for the perturbed system  \eqref{eqn:pert_sys}. Finally, we characterize the size of its domain of attraction
by combining the converse Lyapunov theorems established in \cite{jiang2002converse} 
and techniques similar to the
ones used in the proof of Proposition~\ref{prop:existence_eq}.
Note that the statement of Proposition~\ref{prop:existence_eq}
cannot be employed straightforwardly in the step 4 because 
of the topological obstructions highlighted in 
Section~\ref{sec_topology}.

\paragraph*{Step 1: Local analysis}
Let $\P$ be a positive definite symmetric matrix and $a\in(0,1)$  a real scalar satisfying
\begin{equation*}
    \dfrac{\partial f^\top}{\partial x}(0) \P \dfrac{\partial f}{\partial x}(0) \preceq a \P.
\end{equation*}
Since the origin is a locally exponentially stable equilibrium for system \eqref{eqn:sys}, its linearization around the origin is stable. Hence, the existence of $\P$ is guaranteed by discrete-time Lyapunov inequality. 
In the following, given any $c>0$, we denote with $\cV{c}{\P}$
the subset of $\RR^n$ defined as
\begin{equation}
    \cV{c}{\P}:= \{x\in \RR^n: x^\top \P x \leq c \} \subset\RR^n.
\end{equation}
Now, note that the quadratic Lyapunov function defined by $\P$ is a local Lyapunov function for system \eqref{eqn:sys}. Indeed, by continuity there exists a real number $\varepsilon>0$ such that $\cV{\varepsilon}{\P} \subseteq \lowC $ and
\begin{equation*}
    \dfrac{\partial f^\top}{\partial x}(x) \P \dfrac{\partial f}{\partial x}(x) \preceq \dfrac{1+a}{2} \P
    \qquad\forall x\in \cV{\varepsilon}{\P}.
\end{equation*}
Equivalently (refer to Section~\ref{sec:notation}) it holds
\begin{equation}\label{eqn:contraction_lmi_f}
   \Psi(x) :=
    \begin{pmatrix}
    -\frac{1+a}{2}\P &  \dfrac{\partial f^\top}{\partial x}(x) \P\\
    \P \dfrac{\partial f}{\partial x}(x) & -\P
    \end{pmatrix}
    \preceq 0 
\end{equation}
for all $x\in \cV{\varepsilon}{\P}$.
Consider the candidate Lyapunov function 
$$
V(x)= x^\top \P x, \qquad V(x^+)=f(x)^\top \P f(x).
$$
By defining  a function $F:\RR\to\RR^n$ as
$$
F(s):= f(sx),
$$
and since we assumed the origin to be an equilibrium point for $f$, we have
\begin{align*}
    f(x)= f(x)-f(0)= F(1)-F(0)= \int_0^1 \dfrac{\partial f }{\partial x}( s x) x \;ds.
\end{align*}
Then,
recalling that  $\frac{1+a}{2}<1$, 
we compute
\begin{align} \notag
&V(x^+)-\tfrac{1+a}{2}V(x) = -\frac{1+a}{2} x^\top \P x   \\ \notag
    &\qquad + 2 \int_0^1 x^\top \dfrac{\partial f^\top}{\partial x}(sx) \P f(x) \;ds 
    - f(x)^\top \P f(x) \\ \notag
     & = -\frac{1+a}{2}x^\top \P x \int_0^1  ds +  \int_0^1 x^\top \dfrac{\partial f^\top}{\partial x}(sx) \P f(x) \;ds \\ \notag
    & \qquad +    f(x)^\top \P \int_0^1  \dfrac{\partial f}{\partial x}(sx) x \;ds - f(x)^\top \P f(x)\int_0^1  ds \\ \notag
    & = \int_0^1 \Bigg[ -\frac{1+a}{2} x^\top \P x +  x^\top \dfrac{\partial f^\top}{\partial x}(sx) \P f(x) \;ds \\ \notag
    & \qquad +    f(x)^\top \P \dfrac{\partial f}{\partial x}(sx) x  - f(x)^\top \P f(x) \Bigg]\; ds
    \\ \label{eqn:lyap_err_f}
    & = 
\begin{pmatrix}
x^\top & f(x)^\top
\end{pmatrix}
\int_0^1\Psi(sx) 
ds
\begin{pmatrix}
x \\ f(x)
\end{pmatrix}
\leq 0
\end{align} 
for all $x\in \cV{\varepsilon}{\P}$
and $s\in[0,1]$,
where in the last step we used the definition of $\Psi$
in \eqref{eqn:contraction_lmi_f}.

\paragraph*{Step 2: Existence of an equilibrium for the perturbed system} Now, let $\hat f$ satisfy
\begin{equation} \label{eqn:model_bound}
|\hat f(x)-  f(x)|\leq  \delta_1 := \sqrt{\dfrac{\varepsilon (1-a)^2}{8\lambda_{max}(\P)(3+a)}},
\end{equation}
for all $x\in \cV{\frac\varepsilon2}{\P}$ . In the following, we show that the existence of such a bound and the local stability of system \eqref{eqn:sys} imply the existence of an equilibrium point $x_e \in \cV{\frac\varepsilon2}{\P}$ for the perturbed system \eqref{eqn:pert_sys}. To this aim, we leverage on Brouwer's fixed point theorem. Hence, we need to show that the set  $\cV{\frac\varepsilon2}{\P}$  is forward invariant for  system \eqref{eqn:pert_sys} and that it is homeomorphic to a unit ball. 

Let us choose $\bar x\in \RR^n $ satisfying $\bar x \in \partial \cV{ \frac{\varepsilon}{2}}{\P}$ and note that for all $ x \in \partial \cV{ \frac{\varepsilon}{2}}{\P}$
\begin{multline*}
      \hat f(x)^\top \P \hat f(x) = 
     \left[\hat f(x) - f(x) \right]^\top \P \left[\hat f(x) - f(x) \right] \\ 
     + f(x)^\top \P f(x) + 2 \left[\hat f(x) - f(x) \right]^\top \P f(x) .
\end{multline*}
Then, using the generalized Young's inequality 
$2\alpha \beta \leq \nu^{-1} \alpha^2 + \nu \beta^2$ with $\nu = \dfrac{1-a}{2(1+a)}$ on the last term,
we obtain
\begin{align*}
      \hat f(x)^\top \P \hat f(x) & \leq   
  \dfrac{3+a}{1-a}\left[\hat f(x) - f(x) \right]^\top\P \left[\hat f(x) - f(x) \right] 
  \\  & \qquad + \dfrac{3+a}{2(1+a)} f(x)^\top \P f(x)
  \\
  & \le \dfrac{1-a}{4} {\bar x}^\top \P {\bar x} + \dfrac{3+a}{4} x^\top \P x\\
    & \le  \dfrac{1-a}{4}{\bar x}^\top \P {\bar x} + \dfrac{3+a}{4} {\bar x}^\top \P {\bar x}\le {\bar x}^\top \P {\bar x},
\end{align*}
where we used inequality \eqref{eqn:model_bound}
in the second step.

Thus,  $\cV{\frac\varepsilon2}{\P}$  is forward invariant for the  system \eqref{eqn:pert_sys}. Moreover, since $\partial \cV{\frac{\varepsilon}{2}}{\P}$ is the level surface of a quadratic Lyapunov function, it is homeomorphic to a sphere. Hence, the set $\cV{\frac\varepsilon2}{\P}$ is homomorphic to a closed unit ball. 
Following the proof of Proposition~\ref{prop:existence_eq} and employing Brouwer's fixed point theorem, there exists a point $x_e\in\cV{\frac{\varepsilon}{2}}{\P}$ satisfying
\begin{equation*}
    \hat f(x_e)=x_e.
\end{equation*}
\paragraph*{Step 3: Local exponential stability}
Now we show that $x_e \in \cV{\frac\varepsilon2}{\P}$ is locally exponentially stable for \eqref{eqn:pert_sys}. 
In particular, we show that if $\hat f$  satisfies 
\begin{equation} \label{eqn:jacob_bound}
\left|
\dfrac{\partial \hat f}{\partial x} (x)
-
\dfrac{ \partial  f}{\partial x} (x)
\right| \leq \delta_2 :=\dfrac{1-a}{2\sqrt{10+6a}} ,
\end{equation}
for all $x \in\cV{\varepsilon}{\P}$, then, 
by denoting 
\begin{align*}
    \tilde x &:= x-x_e
    \\
 \bar f(\tilde x) & := \hat f(\tilde x + x_e)-\hat f(x_e)=\hat f(x) -x_e,
\end{align*}
the candidate Lyapunov function
$\widehat V(\tilde x) = \tilde x^\top \P \tilde x$
has to satisfy
\begin{align}\label{eqn:lyap_Vtilde+}
    \widehat V({\tilde x}^+)  = 
   {\bar f(\tilde x) }^\top \P {\bar f(\tilde x)} 
    \le \dfrac{3+a}{4} \widehat V(\tilde x)\le \widehat V(\tilde x) 
\end{align}
for all $\tilde x\in \widetilde {\mathcal{V}}_\varepsilon (\P)$, where 
$$
{\widetilde {\mathcal{V}}_\varepsilon (\P)} = \{\tilde  x\in \RR^n: (x_e+s \tilde x )\in \cV{\varepsilon}{\P},\, \forall s\in [0,1] \}
$$
To show such a property, first note that by defining a function $G:\RR\to\RR^n$ as
$G(s)= \hat f(x_e+s\tilde x)$,
it holds
\begin{align*}
{\bar f(\tilde x)} &= G(1)-G(0) = \int_0^1 \dfrac{\partial \hat f}{\partial x}(x_e+s\tilde x)\; ds\;\tilde x.
\end{align*}
 Then, similarly to the previous part of the proof, we compute
\begin{align} \notag
    &\widehat V({\tilde x}^+)-\dfrac{3+a}{4} V({\tilde x})  = -\frac{3+a}{4} {\tilde x}^\top \P {\tilde x}  - {\bar f(\tilde x)}^\top \P \bar f(\tilde x) \\ \notag
    & \qquad +  2  {\bar f(\tilde x)}^\top \P \int_0^1  \dfrac{\partial \hat f}{\partial x}(x_e+s\tilde x) \tilde x \;ds \\ 
    & = \label{eqn:lyap_error_fhat}
\begin{pmatrix}
{\tilde x}^\top & {\bar f(\tilde x)}^\top
\end{pmatrix}
\int_0^1\widehat \Psi(x_e+s\tilde x) 
ds
\begin{pmatrix}
 {\tilde x} \\ \bar f(\tilde x)
\end{pmatrix},
\end{align} 
where in the last step we defined 
\begin{equation}
   \widehat \Psi(x) :=
    \begin{pmatrix}
    -\frac{3+a}{4}\P &  \dfrac{\partial {\hat f}^\top}{\partial x}(x) \P\\
    \P \dfrac{\partial \hat f}{\partial x}(x) & -\P
    \end{pmatrix}.
\end{equation}
 Recalling Section~\ref{sec:notation} and since $a\in(0,1)$, if
\begin{align*}
 \Phi(x) := \begin{pmatrix}
    -\frac{5+3a}{8}\P & \;\; &\dfrac{ \partial \hat f^\top}{\partial x}(x) \P\\
    \P \dfrac{ \partial \hat f}{\partial x}(x) && -\frac{2(3+a)}{5+3a}\P
    \end{pmatrix}
    \preceq 0
\end{align*}
then  $\widehat \Psi(x)\preceq 0$.
Thus, we can study the sign semi-definiteness of $\Phi(x)$ to conclude the sign semi-definiteness of $\widehat \Psi(x)$. By adding and subtracting $\Psi(x)$ to $\Phi(x)$ and recalling inequality \eqref{eqn:contraction_lmi_f},  we obtain, 
$$
\Phi(x)=\Phi(x)-\Psi(x)+\Psi(x)\preceq \Phi(x)-\Psi(x)
$$
 Then,
\begin{align*}
& \Phi(x) 
    \preceq
    \begin{pmatrix}
    \dfrac{a-1}{8} \P &\; &\dfrac{\partial \tilde f^\top}{\partial x}(x)\P\\
    \P \dfrac{\partial \tilde f}{\partial x}(x)   &  & \dfrac{a-1}{5+3a} \P
    \end{pmatrix}.
\end{align*}
with 
$$
\dfrac{\partial \tilde f}{\partial x}(x) :=
\left[\dfrac{ \partial \hat f}{\partial x} (x)-\dfrac{ \partial f}{\partial x}(x) \right].
$$
 By generalized Schur's complement and bound \eqref{eqn:jacob_bound}, since $a\in(0,1)$, the matrix on the right-hand side is negative semi-definite.
Consequently, it holds
\begin{align*}
   \widehat \Psi (x)  &\preceq 0, \qquad \forall x  \in \Omega_{\varepsilon}.
\end{align*}
As a consequence, we conclude
\begin{align*}
   \widehat \Psi (x_e+s\tilde x)\preceq 0
   \qquad \forall \tilde x\in \widetilde \Omega,
\end{align*}
showing  
\eqref{eqn:lyap_Vtilde+}. 
We proved that $x_e$ is 
locally exponentially stable 
for \eqref{eqn:pert_sys} with a domain of attraction including the set $\{\tilde x\in \widetilde {\mathcal{V}}_\varepsilon (\P)\}$.
Note that by the definition of $\widetilde {\mathcal{V}}_\varepsilon (\P)$ and since $\cV{\varepsilon}{\P}$ is convex, the point $x_e+s \tilde x = (1-s)x_e+sx$ belongs to the set $\cV{\varepsilon}{\P}$ for all $(x, x_e,s)\in \cV{\varepsilon}{\P}\times\cV{\varepsilon}{\P}\times [0,1]$ . Hence, the domain of attraction of $x_e\in \cV{\frac{\varepsilon}{2}}{\P}$ contains the set $\{x\in\cV{\varepsilon}{\P}\}$.

\paragraph*{Step 4: Domain of attraction}
We now provide a stronger lower bound on the size of the domain of attraction of the equilibrium point $x_e$ for the perturbed system \eqref{eqn:pert_sys}. In particular, 
we show that $x_e$ has a domain of attraction that includes the set $\upC$. 

First, by picking 
\begin{equation} \label{eqn:delta_local}
    \delta_3=\min\{\delta_1,\delta_2\}
\end{equation} where $\delta_1$ comes from \eqref{eqn:model_bound} and $\delta_2$ is defined in \eqref{eqn:jacob_bound}, we obtain that the point $x_e$ is a locally exponentially stable equilibrium point for \eqref{eqn:pert_sys}. Moreover, it is contained in $\lowC$ and its domain of attraction includes $\cV{\varepsilon}{\P}$. 

Now, since $\upC$ is forward invariant and since  system \eqref{eqn:sys} is time-invariant and described by a continuous function on $\cA$, we can leverage on \cite[Theorem 1]{jiang2002converse} to claim the existence of smooth Lyapunov functions $V_0:\cA\to\RR_{\ge 0}$ and $V_{\upC}: \cA\to\RR_{\ge 0}$ such that for all $x \in \cA$ it holds
\begin{align*}
    &\alpha_1(|x|)\le V_0(x) \le \alpha_2(|x|),\\
    &V_0(f(x)) \le \rho_0 V_0(x), \\
    &V_0(x)=0 \iff x =0,
\end{align*}
and 
\begin{align*}
    &\alpha_1(\dis(x,\upC))\le V_{\upC}(x)\le \alpha_2(\dis(x,\upC)) ,\\
    &V_{\upC}(f(x)) \le \rho_{\upC}  V_{\upC}(x),\\
    &V_{\upC}(x)=0 \iff x \in  \upC,
\end{align*}
 where $\alpha_1,\alpha_2\in \mathcal{K}_\infty$ and $(\rho_0, \rho_\upC) \in (0,1)\times (0,1)$.  
Note that the two Lyapunov functions can be bounded by different $\mathcal{K}_\infty$ functions. Yet, we can select the minimum (maximum) between those and obtain the same lower (upper) bound for the two. 
Consider now the set $\upC$. Since it is compact and included in $\cA$, there exists a strictly positive real number $\bar d$ such that the set $\cD:=\{x\in \cA:\dis(x,\upC)\in (0, \bar d ] \}$ is a subset of $\cA$, namely $\cD\subsetneqq \cA$.  Let
\begin{equation}\label{eqn:comb_lyap}
\bV(x)=  V_{\upC}(x) + \sigma V_0(x),
\end{equation}
with
\begin{equation*}
    \sigma = \dfrac{\alpha_1(\bar d)}{2\nu}, \qquad \nu = \sup_{x\in \cA:\dis(x,\upC) \le  \bar d } V_0(x).
\end{equation*}
By picking $\rho=\max\{\rho_0, \rho_\upC\}$, for all $x \in \cA$ it satisfies
\begin{align}
    &\alpha_3(|x|) \leq \bV(x) \leq \alpha_4(|x|), \notag \\ \label{eqn:comp_lyap_decr}
    &\bV(f(x)) \le \rho \bV(x),\\
    &\bV(x)=0 \iff x=0 \notag
\end{align}
with $\alpha_3,\alpha_4\in \mathcal{K}_\infty$. This implies that $\bV$ is a Lyapunov function for system \eqref{eqn:sys} on $\cA$.
Define the following sets
\begin{equation}
    \label{eq:setV}
    \begin{array}{rcl}
\partial  \cV{\bar d}{\bV} &:=& \{x \in \cA: \bV(x)=\alpha_1(\bar d)\}, 
  \\
 \cV{\bar d}{\bV} &:=& \{x \in \cA: \bV(x)\leq\alpha_1(\bar d)\}.
  \\ 
    \end{array}
\end{equation}
Due to the definitions of $V_\upC$ in \eqref{eqn:comb_lyap} and its lower bound,  we have
\begin{equation*}
    \alpha_1(\dis(x,\upC))\le V_{\upC}(x) \le \bV(x).
\end{equation*}
Hence, for all $x\in \cV{\bar d}{\bV}$
\begin{equation*}
    \dis(x,\upC) \le \bar d.
\end{equation*}
As a consequence,  for all $x\in \cV{\bar d}{\bV}$, we have $\nu \ge V_0(x)$ and by \eqref{eqn:comb_lyap} it holds
\begin{equation*}
    V_{\upC}(x) = \alpha_1(\bar d)\left(1-\dfrac{V_0(x)}{2\nu}\right) > 0.
\end{equation*}
This last relations imply
\begin{equation*}
    \upC \subsetneqq \cV{\bar d}{\bV} \subseteq \cD \subsetneqq \cA.
\end{equation*}
Identify by $\underline v>0$ a scalar such that if $\bV(x)\le \underline v$  then  $x\in\cV{\frac{\varepsilon}{2}}{\P}$, namely
\[
x^\top \P x \le \dfrac{\varepsilon}{2}, \quad \forall x\in\cA:\bV(x)\le \underline v.
\]
Define 
\begin{equation}
    \label{eqn:delta_global}
   \delta_4:=\min\left\{
\inf_{(y',y'')\in\partial\cV{\bar d}{\bV}\times \partial \upC} |y'-y''|
,\dfrac{(1-\rho )\underline v}{\displaystyle\sup_{x\in\cV{\bar d}{\bV}} \left| \dfrac{\partial \bV}{\partial x}(x)\right|}\right\}
\end{equation}
and consider \eqref{eqn:dyn_cond} with \eqref{eqn:delta_global}. 
We can define the function $p:\RR_{\geq 0}\to\RR$ as
\begin{equation*}
p(s) =\max_{x\in\upC, |v|=1} \{V(f(x)+sv)-\alpha_1(\bar d)\}.
\end{equation*}
Then, by the first argument of \eqref{eqn:delta_global} and \eqref{eq:setV}, for all $s\in [0,\delta_4]$ we have $p(s)<0$.
Note that for all $x$ in $\upC$ such that $f(x)\neq\hat f(x)$, 
\begin{equation}
    \bV(\hat f(x))- \alpha_1(\bar d)= V(f(x)+sv)-\alpha_1(\bar d)
\end{equation}
with $s=|\hat f(x)- f(x)|$, $v=\frac{\hat f(x)-f(x)}{|\hat f(x)-f(x)|}$. Consequently, 
 for all $x$ in $\upC$, 
\begin{equation}
    \bV(\hat f(x))-\alpha_1(\bar d) \leq p(|\hat f(x)-f(x)|)\leq 0.
\end{equation}
Thus, $\hat f(x)\in \cV{\bar d}{\bV}$ for all $x\in\upC$ and
\begin{equation*}
    f(x)+s(\hat f(x)-f(x)) \in \cV{\bar d}{\bV}, \quad \forall (x, s)\in \upC\times[0,1]. 
\end{equation*}
Then, by picking $s\in[0,1]$ and a function $H:\RR\to\RR$ as 
$H(s)=V(f(x)+s(\hat f(x) -f(x)))$, we have
\begin{align*}
    &\bV(\hat f(x))- \bV(f(x))= H(1)-H(0)\\
    &= \int_0^1 \dfrac{\partial \bV }{\partial x}
    \Big( f(x)+s(\hat f(x) -f(x))\Big)  \;ds \, \big[\hat f(x)-f(x)\big].
    \end{align*}
Consequently, 
for all $x\in \cV{\bar d}{\bV}$,
we obtain
\begin{align*}
    \bV(\hat f(x))&\leq  \bV(f(x)) 
        \\
  &  +   
    \int_0^1 \sup_{x\in \cV{\bar d}{\bV}} 
    \Bigg|\dfrac{\partial \bV}{\partial x}(x)\Bigg | \, ds \, \big|\hat f(x)-f(x)\big|
    \\
    &<\rho \bV(x) 
    + \sup_{x\in \cV{\bar d}{\bV}} \left|\dfrac{\partial \bV}{\partial x}(x)\right|
    \delta 
    \\
      &< \rho \bV(x)+ (1-\rho ) \underline v\,. 
\end{align*}
Since $\rho \in (0,1)$, it holds
\begin{equation} \label{eqn:fwd_inv}
\begin{cases}
    \bV(\hat f(x))<  \bV(x)  & \forall x \in \upC\setminus \cV{\frac{\varepsilon}{2}}{\P}\\
    \bV(\hat f(x))<  \underline v \; & \forall x \in  \cV{\frac{\varepsilon}{2}}{\P}.
\end{cases}
\end{equation}
Therefore, we may pick any function $\hat f$
satisfying conditions \eqref{eqn:dyn_cond}, \eqref{eqn:jcb_cond}
with 
$$\delta < \min\{\delta_3,\delta_4\}.$$
The analysis shows that bound \eqref{eqn:delta_global} ensures trajectories of the perturbed system \eqref{eqn:pert_sys} initialized in $\upC$ converge to $\cV{\frac{\varepsilon}{2}}{\P}$. Moreover, we previously proved $\cV{\varepsilon}{\P}$ is included in the exponentially stable domain of attraction of $x_e$. Hence, $x_e$ is an asymptotically stable equilibrium for \eqref{eqn:pert_sys} and it is locally exponentially stable with domain of attraction including $\upC$. 
\end{proof}

Theorem~\ref{thm:perturbed_stable} states that, in presence of bounded mismatches, stability of an equilibrium is preserved for the perturbed system \eqref{eqn:pert_sys}. This new equilibrium may not coincide with  the original one, yet it is guaranteed to be in its neighborhood.
As a particular example, we may study the case in which 
$x^+ = f(x)+\delta$ with 
$\delta$ being a constant small perturbation.
In such scenario, the equilibrium of the perturbed system slightly shifts (it is actually computed as the solution $x_e$ to  $f(x_e)+\delta=x_e$), and we recover the result proposed in \cite[Theorem 2.7]{lin1994design}. Hence, even if we recall that the proof proposed by the authors is either incorrect or missing some arguments, we believe that the general statement may be correct.

Finally, we highlight that the existence of an equilibrium (although different from the nominal one) may be of great value in some applications. A direct example is the control of systems via integral-action, see e.g., \cite{astolfi2017integral,giaccagli2021sufficient, abdelrahem2018efficient}, as detailed in the next section.

\section{Application to integral action control}\label{sec:regulation}

As stated at the end of the previous section, the results of Theorem~\ref{thm:perturbed_stable} can be of interest to the output regulation problem and, in particular, in the context of integral-action controllers. Such control schemes are commonly employed for (constant) perturbation rejection or (constant) reference tracking. The most trivial examples are PI and PID laws. In the following, we specialize our previous result to design an integral action for discrete-time nonlinear systems.

In particular, consider the controlled extended system
\begin{equation}
    \label{eqn:sys_int}
    \left\{
 \begin{aligned}
       \xi^+ =& \varphi(\xi) + g(\xi,u), \quad y=h(\xi),\\
    z^+ =& z + k(\xi,y)      
 \end{aligned}
\right.
\end{equation}
where
$\xi\in \RR^{q}$ is the state of the controlled plant, $u\in \RR^m$ is the control input, 
$y\in \RR^p$ is the output to be regulated to zero 
and $z\in \RR^p$ represents a discrete-time generalized integral action.
Furthermore, we suppose that 
$h(0)=0$, $g(\xi,0)=0$, and $p\le m$. 
The function  $k$ is a $C^1$ function satisfying
the following set of conditions
\begin{equation}\label{eqn:gen_integr}
\begin{array}{rl}
    k(\xi,y) = 0 &\iff y = 0\\[.5em]
    |k(\xi,y_a)-k(\xi,y_b)| &\le L_1(\xi)|y_a-y_b|\\
    \left|\dfrac{\partial  k}{\partial \xi}(\xi, y_a) -\dfrac{\partial k}{\partial \xi}(\xi, y_b)\right|&\le L_2(\xi) \left| y_a - y_b\right|
\end{array}
\end{equation}
for all $\xi\in \RR^q$, for all $(y_a,y_b) \in \RR^p\times\RR^p$ and for some continuous functions $L_1,L_2:\RR^n\to\RR_{\ge 0}$. It can be easily seen that 
by selecting $k(\xi,y)=y$ 
we recover 
the standard discrete-time integrator
$$
z^+ = z+y
$$
that trivially satisfies the previous requirements with $L_1(\xi):=1, \, L_2(\xi):=0$ for all $\xi$. 
Under a feedback controller $u=\alpha(\xi,z)$, we can define the extended dynamics for $x = (\xi^\top , z^\top)^\top$ as
\begin{equation}\label{eqn:ext_sys}
    x^+ = f(x) := \begin{pmatrix}
    \varphi(\xi) + g(\xi,\alpha(\xi,z)) \\ z + k(\xi,h(\xi))
    \end{pmatrix},
\end{equation}
and the perturbed model
\begin{equation}\label{eqn:pert_ext_sys}
    x^+ = \hat f(x) := \begin{pmatrix}
    \hat \varphi(\xi) + \hat g(\xi,\alpha(\xi,z)) \\ z + k(\xi,\hat h(\xi,\alpha(\xi,z))
    \end{pmatrix}.
\end{equation}
Thus, the goal is to study the robustness properties of the closed-loop $x^+=f(x)$.
To this end, we introduce the following compact notation that will be used in the following sections:
\begin{equation}\label{eq:Delta_notation}
    \begin{aligned}
            \Delta_\xi(x) & :=  |\hat \varphi(\xi)-  \varphi(\xi) + \hat g(\xi,\alpha(\xi,z)) - g(\xi,\alpha(\xi,z))|,\\
        \Delta_y (x) & :=  |\hat h(\xi,\alpha(\xi,z)) - h(\xi)|,
        \\
        \Delta_z(x) & := |k(\xi,\hat(\xi,\alpha(\xi,z))) - 
k(\xi,h(\xi) |,
\\
                \Delta_{\partial \xi}(x) &: = \left|  \dfrac{\partial \hat \varphi}{\partial \xi} (\xi) - \dfrac{ \partial  \varphi}{\partial \xi} (\xi) \right.
                \\ & \left.\qquad +  \dfrac{\partial \hat g}{\partial \xi} (\xi,\alpha(\xi,z))  - \dfrac{ \partial  g}{\partial \xi} (\xi,\alpha(\xi,z))  \right| , 
                \\
                \Delta_{\partial u} (x)& :=  \left|\dfrac{\partial \hat g}{\partial u} (\xi,\alpha(\xi,z)) - \dfrac{ \partial  g}{\partial u} (\xi,\alpha(\xi,z)) \right|,
                \\ 
                \Delta_{\partial y} (x) &: = \left| \dfrac{\partial \hat h}{\partial \xi} (\xi,\alpha(\xi,z)) - \dfrac{ \partial  h}{\partial \xi} (\xi) \right| +  \left|\dfrac{\partial \hat h}{\partial u} (\xi,\alpha(\xi,z))\right|. 
        \end{aligned}
        \end{equation}

 \subsection{Existence of equilibria}
We first study the existence of equilibria where regulation is achieved for the perturbed system \eqref{eqn:pert_ext_sys}. Hence, we assume $\varphi, \, g, \, \alpha, \, h$ to be continuous functions.
\begin{proposition} \label{prop:existence_eq_ext}
Let Assumption~\ref{asmp:homeo_sphere} hold for the closed-loop  system \eqref{eqn:ext_sys}.
Then, for any positive $\underline c\leq  \bar c$ there exists a  positive real number $\bar \delta>0$ such that, for any  continuous functions $\hat \varphi:\RR^q\to\RR^q,\, \hat g:\RR^q\times\RR^m\to\RR^q,\, \hat h:\RR^q\times\RR^m\to\RR^p$ 
satisfying
\begin{equation} \label{eqn:dyn_cond_int_ext}
    \Delta_\xi(x) + \Delta_y(x) \leq \bar \delta,
    \qquad \forall\, x\in\cV{\bar c}{V},
\end{equation}
with $\Delta_\xi,\Delta_y$ defined in \eqref{eq:Delta_notation},
the corresponding
system \eqref{eqn:pert_ext_sys}
 admits an equilibrium point $(\xi_e,z_e) \in \cV{\underline c}{V}$ on which output
 regulation is achieved, i.e. 
  $$\hat h(\xi_e,\alpha(\xi_e,z_e))=0.$$
Furthermore, system \eqref{eqn:pert_ext_sys}
has no other 
 equilibrium in the set
 $\cV{\bar c}{V}\setminus\cV{\underline c}{V}$.
\end{proposition}

\begin{proof}
The proof relies on Proposition~\ref{prop:existence_eq}.  It can be easily checked that \eqref{eqn:dyn_cond_int_ext} implies \eqref{eqn:dyn_cond} for the extended system \eqref{eqn:ext_sys}. Since $\RR^q\times\RR^p$ is a finite dimensional space, all norms are equivalent. Hence, there exist a strictly positive constant $\ell$ such that $|x|\leq \ell \sum_{i=1}^n |x_i|$
where $x_i$ denotes the $i$-th component 
of $x$.
Hence, recalling the compact notation 
\eqref{eq:Delta_notation}, 
we have
\begin{align*}
    |\hat f(x)-f(x)|& \le \ell (\Delta_\xi(x) + \Delta_z(x))
\end{align*}
for all $x\in \RR^{q+p}$.
By defining 
$$
L := \sup_{\xi\in\cV{\bar c}{V}} L_1(\xi) \ge 0,
$$
we obtain 
$|\Delta_z(x)| \leq L |\Delta_y(x)|$ for all
$x\in \cV{\bar c}{V}$
and therefore
equation \eqref{eqn:gen_integr} yields
\begin{align*}
        |\hat f(x)-f(x)| \leq \ell (1+L)  (\Delta_\xi(x) + \Delta_y(x))
\end{align*}
for all $x\in \cV{\bar c}{V}$.
Hence, \eqref{eqn:dyn_cond_int_ext} implies \eqref{eq:prop1_condition} if 
$$
\bar \delta < \dfrac{\delta}{\ell(1+L)}.
$$
Then, Proposition~\ref{prop:existence_eq} guarantees the existence of an equilibrium $(\xi_e,z_e)\in \cV{\underline c}{V}.$ In such an equilibrium, $z_e^+=z_e$ and consequently
$k(\xi,\hat h(\xi_e,\alpha(\xi_e,z_e))) = 0$.
By \eqref{eqn:gen_integr},  the last relation implies $\hat h(\xi_e,\alpha(\xi_e,z_e)) = 0$ and this concludes the proof.
\end{proof}

Proposition~\ref{prop:existence_eq_ext} is a direct application of  Proposition~\ref{prop:existence_eq} to the extended system \eqref{eqn:ext_sys}. Hence, it guarantees the existence of (at least one) equilibrium for the perturbed closed-loop dynamics \eqref{eqn:pert_ext_sys}. In turn, due to the presence of the integral action, the outputs $y$ are still regulated to $0$, even in presence of model mismatches. Thus, Proposition~\ref{prop:existence_eq_ext} shows that the addition of discrete-time generalized integrator dynamics to discrete-time nonlinear systems guarantees the existence of equilibria where (constant) reference tracking and (constant) disturbance rejection is achieved. Moreover, it shows that controllers designed to generate equilibrium points for the nominal extended closed-loop effectively guarantee the existence of equilibria for the perturbed dynamics.

 \subsection{Robust regulation}
Proposition~\ref{prop:existence_eq_ext} does not provide any characterization of the attractivity properties of the equilibrium and hence of the regulation objective $y\to0$. As a matter of fact, the generation of attractive and stable equilibria is often the main objective of control designs. Hence, we build on Theorem~\ref{thm:perturbed_stable} to show that the addition of a discrete-time integral component to stabilizing controllers allows for robust asymptotic (constant) reference tracking and (constant) disturbance rejection.

\begin{assumption} \label{asmp:open_loop_stable}
There exists an open set $\cA\subseteq\RR^q\times \RR^p$ and a $C^1$ function $\alpha:\RR^q\times\RR^p\to\RR^m$ such that the control law $u=\alpha(\xi,z)$ makes the origin of system \eqref{eqn:ext_sys} asymptotically stable with domain of attraction $\cA$ and locally exponentially stable.
\end{assumption}
 
Under Assumption~\ref{asmp:open_loop_stable}, we can infer robust setpoint-tracking properties for the extended system \eqref{eqn:ext_sys} by means of the results in the previous section.

\begin{proposition}\label{prop:robust_integral}
Let Assumption~\ref{asmp:open_loop_stable} hold and let $\lowC$ be an arbitrary compact subset of $\cA$ including the origin. Moreover, let functions $\varphi,g, h$ be
 $C^0$ for all $(\xi,z)\in \cA$
and moreover  
 $C^1$ for all $(\xi,z)\in \lowC$. Then, for any forward invariant   compact set $\upC$ verifying
$$
\{0\} \subsetneqq   \lowC \subsetneqq   \upC  \subsetneqq  \cA,
$$
there exist a strictly positive scalar $\bar\delta$ such that, for 
 functions $\hat \varphi, \hat g, \hat h$ 
 that are
 $C^0$ for all $(\xi,z)\in \cA$
and   
 $C^1$ for all $(\xi,z)\in \lowC$, 
 and moreover satisfy
 \begin{align}
  \Delta_\xi(x) + \Delta_y (x) & \leq \bar \delta,     &\forall x\in \upC, 
  \label{eqn:dyn_cond_int} 
  \\
  \label{eqn:jcb_cond_int}
     \Delta_{\partial \xi}(x) + \Delta_{\partial u} (x) + \Delta_{\partial y} (x) &\le \bar \delta,
      & \forall\, x\in \lowC,
 \end{align}
 with $ \Delta_\xi, \Delta_y, \Delta_{\partial \xi}, \Delta_{\partial u}, \Delta_{\partial y}$
defined in \eqref{eq:Delta_notation},
the corresponding system 
 \eqref{eqn:pert_ext_sys}
admits an equilibrium point $(\xi_e,z_e)\in \lowC$  which is locally exponentially stable and asymptotically stable with domain of attraction containing $\upC$. Moreover, on such equilibrium,   $\hat h(\xi_e,\alpha(\xi_e,z_e))=0$.
\end{proposition}

\begin{proof}
To prove the existence of an asymptotically stable equilibrium for \eqref{eqn:pert_ext_sys} with a locally exponential behavior, we rely on Theorem~\ref{thm:perturbed_stable}. As in the proof of Proposition~\ref{prop:existence_eq_ext}, it can be easily checked that \eqref{eqn:dyn_cond_int} implies \eqref{eqn:dyn_cond} for the extended system \eqref{eqn:ext_sys}. By following the same steps, we define
$$
L = \sup_{\xi\in\upC} L_1(\xi) \ge 0.
$$
With the same definition of $\ell$ as in the proof of Proposition~\ref{prop:existence_eq_ext}, \eqref{eqn:dyn_cond_int} implies \eqref{eqn:dyn_cond} for al $x\in\upC$ if 
$$
\delta < \dfrac{\min\{\delta_3,\delta_4\}}{\ell(1+L)},
$$
with $\delta_3,\delta_4$ defined in Step 4 of the proof of Theorem~\ref{thm:perturbed_stable}.

Similarly, we can check that \eqref{eqn:jcb_cond_int} implies \eqref{eqn:jcb_cond} for the extended system \eqref{eqn:ext_sys}. Let us define  
\begin{align*}
    L_\alpha&=\sup_{(\xi,z)\in \lowC} \left\{\Bigg|\dfrac{\partial \alpha}{\partial \xi}(\xi,z)\Bigg|,\Bigg|\dfrac{\partial \alpha}{\partial z}(\xi,z)\Bigg|\right\}\\
    L_k&=\sup_{\xi\in \lowC} \left\{L_1(\xi), L_2(\xi)\right\}\ge L.
\end{align*} 
To show this, let $\bar \ell>0$ such that
$|A|\leq \bar\ell \sum_{i=1}^n |A_i|$
where $A$ is any (rectangular) matrix and $A_i$
denotes the $i$-th row.
Then, 
by recalling the notation introduced in 
\eqref{eq:Delta_notation} and using the previous inequalities, 
we obtain
\begin{align*}
    \Bigg|\dfrac{\partial \hat f}{\partial x}&(x)  - \dfrac{\partial  f}{\partial x}(x)\Bigg| 
    \\
    & \le \bar\ell  \Delta_{\partial \xi}(x) 
    + 2\bar \ell L_\alpha  \Delta_{\partial u}(x) 
    +2 \bar\ell L_\alpha L_k\Bigg| \dfrac{\partial \hat h}{\partial u}(\xi,\alpha(\xi,z))\Bigg|
    \\& 
    \qquad + 2\bar\ell L_k\Bigg|\dfrac{\partial \hat h}{\partial \xi}(\xi,\alpha(\xi,z)) - \dfrac{\partial  h}{\partial \xi}(\xi)  
    \Bigg|
    \\ 
    &\le \bar \ell(1+2(L_\alpha+L_k+L_\alpha L_k)) (\Delta_{\partial \xi} +\Delta_{\partial u}  +\Delta_{\partial y})
\end{align*}
for all $x\in\lowC$.
Then, by picking
\begin{equation*}
    \mu = \max \left\{\ell (1+L),\bar\ell(1+2(L_\alpha+L_k+L_\alpha L_k))\right\}
\end{equation*}
and 
$$
\bar \delta < \dfrac{\min\{\delta_3,\delta_4\}}{\mu}.
$$
Theorem~\ref{thm:perturbed_stable} guarantees the existence of an  asymptotically stable equilibrium $x_e = (\xi_e,z_e)$ which is close to the origin and locally exponentially stable. In such an equilibrium, $z_e^+=z_e$ and consequently
$k(\xi,\hat h(\xi_e,\alpha(\xi_e,z_e))) = 0$.
By \eqref{eqn:gen_integr}, $\hat h(\xi_e,\alpha(\xi_e,z_e)) = 0$ and this concludes the proof.
\end{proof}
Proposition~\ref{prop:robust_integral} states that discrete-time  stabilizing controllers exploiting integral actions achieve robust output regulation to constant setpoints. As an intuitive example, we may consider controllers designed for a nonlinear model, whose constant parameters are identified through data. Then, Proposition~\ref{prop:robust_integral} guarantees that, if the approximation error is sufficiently small, the controller will still stabilize the system and the output will reach the desired constant setpoint.

\subsection{About Assumption~\ref{asmp:open_loop_stable} and Forwarding approach}
It is not straightforward to see when and how Assumption~\ref{asmp:open_loop_stable} can be satisfied, namely
how to design the feedback law 
$\alpha(x,z)$ for the extended dynamics
\eqref{eqn:sys_int}. An elegant solution comes from discrete-time forwarding techniques\cite{monaco2013stabilization,mattioni2019forwarding}. If the origin of the autonomous system $\xi^+=\varphi(\xi)$ is globally asymptotycally stable and locally exponentially stable, the origin of the extended system \eqref{eqn:ext_sys} can be stabilized via a feedback controller. 
By considering the SISO scenario $m=p=1$, the result in \cite[Theorem 3.1]{mattioni2017lyapunov} provides a feedback law $u = \alpha(\xi,z)$ achieving global asymptotic and local exponential stability of the origin of the extended system \eqref{eqn:ext_sys}.

In particular, suppose that 
$W:\RR^q\to\RR_{\ge 0}$ is a Lyapunov function for system \eqref{eqn:sys} which is locally quadratic. Assume also that the linearization of \eqref{eqn:ext_sys} around the origin is stabilizable and suppose 
to know 
a mapping $M:\RR^q\to\RR$ satisfying
\begin{equation*}
    M(\xi^+) = M(\xi) + k(\xi,h(\xi)).
\end{equation*}
 Then, we can perform a change of coordinates $\eta = z - M(\xi)$ and consider $\zeta = (\xi^\top , \eta^\top)^\top$. This leads to the stabilizing controller $u=\bar \alpha(\zeta)$ 
given by the (implicit) solution to
\begin{equation*}
    u=-\dfrac{1}{u}\int_0^u \dfrac{\partial V}{\partial \zeta}G(\zeta^+(v),v)dv
\end{equation*}
with 
\begin{align*}
    V(\zeta) &= W(\xi) + \eta^2 \\
    G(\zeta,u)^\top &= \left( \dfrac{\partial  g}{\partial u}(\xi,u), -\dfrac{\partial M}{\partial \xi}\dfrac{\partial  g}{\partial u}(\xi,u)\right).
\end{align*}
In the original coordinates, we obtain 
$$
\alpha(\xi,z) = \bar \alpha(\xi, z- M(\xi)).
$$
The function $V$ above defined is a Lyapunov function 
for the closed-loop systems. 
Moreover, note that if the function $W$ has 
the desired homeomorphicity properties (see in particular Section~\ref{sec_topology}), then so has $V$ by construction.

The aforementioned methodology can be generalized to the MIMO scenario following techniques similar to \cite[Section 6]{monaco2011nonlinear}. For discussion about the existence of such a mapping $M$, refer to \cite{mattioni2017lyapunov}.
 In the continuous time case we find an equivalent characterization under the so-called non-resonance conditions, see, e.g. \cite{astolfi2017integral}.

\section{Conclusion} \label{sec:conclusion}
In this paper, we proved that stability properties of a discrete-time nonlinear system transfer to plants described by sufficiently similar models. The results are independent of the structural properties of the system. We applied these results to the robust output regulation problem, showing the rejection of constant disturbances via integral action for nonlinear systems. The application to the problem of robust output regulation in the presence of more complex exogenous signals is  subject of ongoing research.

\medskip
\noindent
\textbf{Acknowledgement.}
The authors thank Lucas Brivadis for 
providing the example \eqref{eq_SystLyap} and its associated Lyapunov function, whose sublevel sets are not homeomorphic to a ball.


    

\bibliography{biblio.bib}

\end{document}